\documentclass[letterpaper, 10 pt, conference]{ieeeconf}

\IEEEoverridecommandlockouts

\usepackage{graphicx} 
\usepackage{cite}

\usepackage{amsmath,amssymb,amsfonts}

\usepackage{amsthm}
\usepackage{tabularx}
\usepackage{algorithmic}
\usepackage{graphicx}
\usepackage{textcomp}
\usepackage{xcolor}
\usepackage{url}

\usepackage{caption}
\usepackage{subfigure}
\usepackage{dsfont}
\usepackage{mathtools}
\usepackage{threeparttable}
\usepackage{colortbl}
\usepackage{siunitx}
\usepackage{multirow}
\usepackage{algorithm}

\allowdisplaybreaks

\DeclareSymbolFont{bbold}{U}{bbold}{m}{n}
\DeclareSymbolFontAlphabet{\mathbbold}{bbold}

\newcommand{\diag}[1]{\ensuremath{\mathrm{diag}(#1)}}
\newcommand{\tr}[1]{\ensuremath{\mathrm{tr}(#1)}}
\newcommand{\rank}[1]{\ensuremath{\mathrm{rank}(#1)}}

\theoremstyle{bfnote}
\newtheorem{thm}{Theorem}

\newtheorem{rem}{Remark}

\DeclareSIUnit[]{\pu}{p.u.}
\DeclareSIUnit[]{\VA}{VA}
\newcommand{\real}[0]{\mathbb R}

\usepackage{setspace}
\title{\LARGE \bf
Optimal Allocation of Grid-Forming Frequency Shaping Control
}

\author{Zhaomin Lyu$^{1}$, Ding Zhang$^{2}$, and Yan Jiang$^{1}$
\thanks{This work was supported in part by Guangdong Association for Science \& Technology Young Elite Scientist Cultivation Program SKXRC2026403 and CUHKSZ University Development Fund 01003790.
\textit{(Corresponding author: Yan Jiang.)}}
\thanks{$^{1}$Z. Lyu and Y. Jiang are with the School of Science and Engineering, The Chinese University of Hong Kong, Shenzhen, 518172, CHN. Emails: {\tt \{zhaominlyu,yjiang\}@cuhk.edu.cn}}%
\thanks{$^{2}$D. Zhang is with the School of Engineering, Australian National University, Canberra, AUS. Email: {\tt ding.zhang@connect.ust.hk}}
}

\begin{document}

\maketitle

\thispagestyle{empty}
\pagestyle{empty}

\begin{abstract}
Various inverter-based control strategies have been proposed to improve frequency security for power systems with high renewable penetration. Among them, the grid-forming frequency shaping control is particularly promising due to its ability to shape the post-contingency aggregate system frequency dynamics into first-order with prescribed rate of change of frequency (RoCoF) and steady-state frequency deviation. Moreover, the shaped aggregate dynamics depends on the harmonic sum of all inverter transfer functions rather than on their distribution. In light of this, here we explore how to allocate shaping control resources to minimize the transient control effort needed to achieve ideal coherent dynamics. We formulate it as a constrained optimization problem by tackling two main difficulties. First, we make the squared $\mathcal{H}_2$ norm legitimate quantification of transient control cost under step power imbalances through a system transformation. Second, we simplify the $s$-domain constraint for Nadir elimination into a standard constraint that is easy to implement in optimization. The resulting non-convex optimization problem can be solved by existing successive convexification method. The effectiveness
of the allocation has been verified on the modified Icelandic Power Network test case.

\end{abstract}

\section{Introduction}\label{sec:intro}

The urgent need for decarbonization is driving the generation mix toward cleaner energy resources such as solar and wind, which  are typically interfaced to power grids through power electronic inverters~\cite{benjamin2017}. The resulting loss of physical inertia and damping inherently provided by synchronous generators jeopardizes the grid frequency security~\cite{milano2018}. To address this issue, various frequency control strategies have been proposed for inverter-based resources (IBRs), including but not limited to droop control~\cite{ofir2018DCvsvVI}, virtual inertia~\cite{fang2017distributed,Poolla2017TAC}, and frequency shaping~\cite{jiangtps2021,jiang2021lcss,Jiang2026TPS}. The former two methods are widely adopted in practice, both of which leverage flexible inverter power output to mimic synchronous generator behavior, thereby enhancing frequency security. However, simply sticking to synthetic inertia and droop response cannot take full advantage of IBRs~\cite{jiang2021tac}. An alternative is our recently proposed frequency shaping control which, as the name suggests, is able to shape the aggregate system frequency dynamics following a
sudden power imbalance into a first-order one that naturally
has no Nadir, while simultaneously attaining the pre-specified rate of change of frequency (RoCoF) and steady-state frequency deviation~\cite{jiangtps2021,jiang2021lcss,Jiang2026TPS}. Like other prevalent frequency control strategies, frequency shaping can also be implemented in either grid-following~\cite{jiangtps2021,Jiang2026TPS} or grid-forming~\cite{jiang2021lcss} mode, where the latter responds to power imbalances by directly adjusting frequency, thus avoiding delay introduced by the phase-locked loop and low-pass filter~\cite{jiang2021lcss}. Hence, here, we exclusively focus on the \emph{grid-forming frequency shaping control}, a promising approach for future power systems.

When deploying the grid-forming frequency shaping control, the aggregate coherent dynamics depends purely on the harmonic sum of all inverter transfer functions, irrespective of their distribution across the network~\cite{jiang2021lcss}. This leaves room for exploring the optimal allocation of inverter control resources over the network to minimize control effort required for desired aggregate frequency dynamics, which has not been investigated by~\cite{jiang2021lcss}. In fact, the optimal placement of virtual inertia and droop control based on system norms that characterize frequency control performance metrics has been extensively studied~\cite{Poolla2017TAC, Poolla2019tps}. For example, \cite{Poolla2019tps} proposed to allocate virtual inertia and droop coefficients of inverters to minimize a weighted quadratic cost penalizing frequency excursions and control effort, where the squared $\mathcal{H}_2$ norm is used to quantify the cost under a series of impulse power disturbances. However, contingencies in power systems are typically deemed as step disturbances. Thus, the precision of using the $\mathcal{H}_2$ norm defined in \cite{Poolla2019tps} to measure the cost during the whole transient is questionable. In addition, compared to common virtual inertia and droop control, the optimal allocation of grid-forming frequency shaping control resources must obey an $s$-domain constraint for Nadir elimination~\cite{jiang2021lcss}, which is challenging to enforce directly. Clearly, care must be taken to properly formulate the allocation of grid-forming frequency shaping control resources as a constrained optimization problem.

With this in mind, we address the transient cost quantification and the $s$-domain constraint simplification separately. First, we transform the original system with power disturbances as natural inputs into a system with the time derivative of power disturbances as inputs, which makes the squared $\mathcal{H}_2$ norm of the resulting system a reasonable estimate of the control effort during the transient induced by step power imbalances in the original system. Second, we simplify the $s$-domain constraint into a standard equality constraint that can be readily incorporated into the optimization problem. This enables us to reformulate the grid-forming frequency shaping control allocation problem as a constrained non-convex optimization problem which can be solved by the existing successive convexification algorithm~\cite{Oguri2023cdc}. The effectiveness
of the allocation has been verified on the modified Icelandic Power Network test case.

\section{Modeling Approach and Problem Statement}\label{sec:model}
In this section, we present the power network model under grid-forming frequency shaping control and propose the optimal control resource allocation problem to be solved.
\subsection{Power System Model}\label{sec: 2A}
We consider a power network composed of $n$ buses indexed by $i \in \mathcal{N} := \{1,\dots, n\}$ and transmission lines denoted by unordered pairs $\{i,j\} \in \mathcal{E}\subset \{\{i,j\}:i,j\in\mathcal{N},i\not=j\}$. As illustrated in Fig. \ref{fig:model}, the system dynamics can be modeled as a feedback interconnection of bus dynamics and network dynamics~\cite{pm2019tac,jiang2021tac,jiang2021lcss}.
The input signals $\boldsymbol{p}_\mathrm{in} := \left(p_{\mathrm{in},i}, i \in \mathcal{N} \right) \in \real^n$ (in $\SI{}{\pu}$) represent power injection changes and the output signals $\boldsymbol{\omega}:=\left(\omega_i, i \in \mathcal{N} \right) \in \real^n$ (in $\SI{}{\pu}$) represent the bus frequency deviations from its nominal value. We now elaborate on the dynamic elements.

\begin{figure}
\centering
\includegraphics[width=0.65\columnwidth]{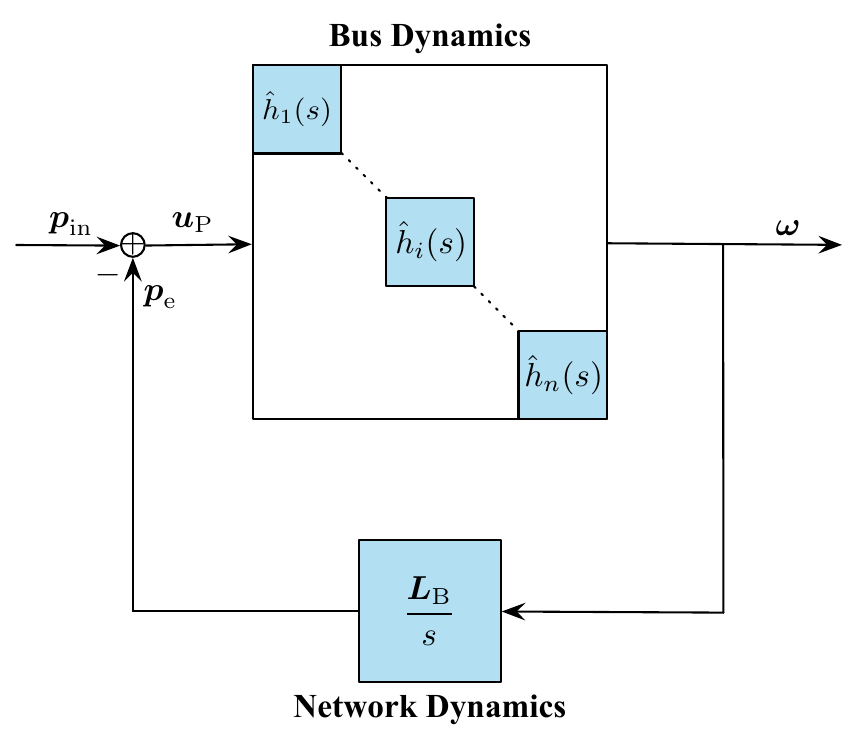}
\caption{Block diagram of power network.}\label{fig:model}
\end{figure}

\subsubsection{Bus Dynamics} We focus on the case where each bus is equipped with either a synchronous generator or a grid-forming inverter. Without loss of generality,
the set of buses $\mathcal{N}$ can be reordered and partitioned as a disjoint union of two subsets: the generator set $\mathcal{G}:=\{1,\ldots,|\mathcal{G}|\}$ and the inverter set $\mathcal{I}:=\{|\mathcal{G}|\!+\!1,\ldots,n\}$, where $|\mathcal{G}|$ denotes the cardinality of $\mathcal{G}$. Depending on whether $i \in \mathcal{G}$ or $i \in \mathcal{I}$, each bus $i$ maps the net power bus imbalance $u_{\mathrm{P},i}$ (in $\SI{}{\pu}$) to the frequency deviation $\omega_i$ according to a transfer function $\hat{h}_i(s)$ described below.\footnote{We use hat to denote the Laplace transform.}  \paragraph{Generator Dynamics}
We consider generator dynamics that are characterized by the standard swing dynamics governed by turbine droop through its mechanical power output variation $q_{\mathrm{t},i}$, i.e.,
\begin{subequations}\label{eq:gen-ode}
  \begin{align}
m_i \dot{\omega}_i =& - d_i \omega_i +q_{\mathrm{t},i} + u_{\mathrm{P},i} \,,\\
\tau_i\dot q_{\mathrm{t},i}=& -{\alpha_{\mathrm{t},i}}\omega_i - q_{\mathrm{t},i}\,,
\end{align}  
\end{subequations}
whose transfer function is 
\begin{align} \label{eq:dy-sw-t}
\hat{h}_i(s) =& \left(m_i s + d_i + \frac{\alpha_{\mathrm{t},i}}{\tau_i s + 1}\right)^{-1}\,,\qquad \forall i \in \mathcal{G}\,.
\end{align}
Here, $m_i>0$ (in $\SI{}{\second}$) denotes the aggregate generator inertia, $d_i>0$ (in $\SI{}{\pu}$) the aggregate generator damping, $\tau_i>0$ (in $\SI{}{\second}$) the turbine time constant, and $\alpha_{\mathrm{t},i}>0$ (in $\SI{}{\pu}$) the turbine inverse droop coefficient.
\paragraph{Inverter Dynamics}

We consider grid-forming inverters for their great potential to address weak grid conditions in future power systems. They are deemed the cornerstone of future grids due to their ability to set local grid frequency deviation $\omega_i$ directly as a function of their power output variation $(-u_{\mathrm{P},i})$. The detailed function depends on the control law $\hat{h}_i(s)$ adopted to map $u_{\mathrm{P},i}$ to $\omega_i$ for buses with $i \in \mathcal{I}$. The prevalent ones are grid-forming droop control and grid-forming virtual inertia, both of which attempt to mimic synchronous generator behavior through inverters. In this manuscript, however, we focus on the grid-forming frequency shaping control~\cite{jiang2021lcss} which is a promising approach in that it is able to shape the aggregate system frequency dynamics into a first-order one with the desired steady-state frequency deviation and RoCoF~\cite[Theorem 1]{jiang2021lcss}. For discussions on the superior performance of frequency shaping control over common approaches, we refer to~\cite{jiang2021lcss}. Specifically, a simple way to achieve grid-forming frequency shaping control is to set~\cite{jiang2021lcss}
\begin{equation}\label{eq:hi-inverter}
\hat{h}_i(s) = \left(m_{i} s + d_{i}-\frac{\rho_i}{\sigma_i s + 1}\right)^{-1}\,,\qquad\forall i \in \mathcal{I}\,,
\end{equation}
where the tunable inverter parameters should satisfy $m_{i}>0$ (in $\SI{}{\second}$), $d_{i} > \rho_i\geq 0$ (in $\SI{}{\pu}$), $\sigma_i>0$ (in $\SI{}{\second}$), and
\begin{align}\label{eq:thm1-condition}
\sum_{i \in \mathcal{I}} \frac{\rho_i}{\sigma_i s + 1} =& \sum_{i \in \mathcal{G}}\frac{\alpha_{\mathrm{t},i}}{\tau_i s + 1}\,.
\end{align}
Moreover, the sums of $m_{i}$ and $d_{i}$ over $i\in\mathcal{I}$ should be designed based on the specified RoCoF and steady-state frequency deviation, respectively, for frequency security, which will be reviewed in Section~\ref{ssec:allocation-pro}.
\subsubsection{Network Dynamics}
The network power fluctuations $\boldsymbol{p}_\mathrm{e} := \left(p_{\mathrm{e},i}, i \in \mathcal{N} \right) \in \real^n$ are captured by a linearized power flow model~\cite{Purchala2005dc-flow}:
\begin{align}
 \hat {\boldsymbol{p}}_\mathrm{e}(s) = \frac{\boldsymbol{L}_\mathrm{B}}{s} \hat {\boldsymbol{\omega}}(s)\;,\label{eq:N}
\end{align}
where $\boldsymbol{L}_\mathrm{B}\in \real^{n\times n}$ is an undirected weighted Laplacian matrix of the transimission network with its $ij$-th element being
\[
L_{\mathrm{B},{ij}}=\Omega_0\partial_{\theta_j}{\sum_{l=1}^n|V_i||V_l|B_{il}\sin(\theta_i-\theta_l)}\Bigr|_{\boldsymbol{\theta}=\boldsymbol{\theta}_0}.
\]
Here, $\boldsymbol{\theta} := \left(\theta_i, i \in \mathcal{N} \right) \in \real^n$ are the voltage angles with $\boldsymbol{\theta}_0$ being the equilibrium angles (in $\SI{}{\radian}$), $|V_i|$ is the (constant) voltage magnitude at bus $i$ (in $\SI{}{\pu}$), $B_{ij}$ is the line $\{i,j\}$ susceptance (in $\SI{}{\pu}$), and $\Omega_0:=2\pi F_0$ is the nominal angular frequency (in $\SI{}{\radian/s}$) with $F_0$ being $\SI{50}{\hertz}$ or $\SI{60}{\hertz}$
depending on a particular system. 
\subsubsection{Closed-Loop Dynamics} 
The frequency deviations $\boldsymbol{\omega}$ in response to the power injection changes $\boldsymbol{p}_\mathrm{in}$ in this system have been shown to be approximately given by 
\begin{equation}\label{eq:sys-fre}
   \!\!\!\hat {\boldsymbol{\omega}}(s)\!\approx\!\frac{\mathbbold{1}_n \mathbbold{1}_n^T\hat{\boldsymbol{p}}_\mathrm{in}(s)}{(\sum_{i \in \mathcal{I}}m_{i} \!+\!\sum_{i \in \mathcal{G}} m_i)s+(\sum_{i \in \mathcal{I}} d_{i} \!+\!\sum_{i \in \mathcal{G}}d_i)} 
\end{equation}
when the network is tightly-connected~\cite{jiang2021lcss}, which is beneficial for frequency security since all buses exhibit coherent first-order dynamics that naturally has no frequency Nadir following sudden power disturbances.

\subsection{Optimal Frequency Shaping Control Allocation Problem}\label{ssec:allocation-pro}
Based on the model just presented, we are particularly interested in the optimal allocation of grid-forming frequency shaping control resources. Precisely, we would like to tune parameters of frequency shaping control over the network to minimize the transient control effort required from both generators and inverters to meet frequency security specification following major disturbances. 

To formally define this problem, we consider the case where the system in Fig.~\ref{fig:model} experiences sudden step power disturbances $\boldsymbol{p}_\mathrm{in} =\boldsymbol{u}_0 \mathds{U}_{ t \geq 0 }$ with $\boldsymbol{u}_0 := \left(u_{0,i}, i \in \mathcal{N} \right) \in \real^n$ being an arbitrary vector direction that allows for power disturbances of different magnitudes at individual buses and $\mathds{U}_{ t \geq 0 } $ representing the unit-step function. Supposing that the expected maximum magnitude of the net power imbalance that the system should survive is $|\sum_{i=1}^n u_{0,i}|=\Delta P$ (in $\SI{}{\pu}$), \cite[Theorem 1]{jiang2021lcss} suggests to tune inverter parameters such that 
\begin{align}\label{eq:I_parameters}
  \sum_{i \in \mathcal{I}}m_{i} \!+\!\sum_{i \in \mathcal{G}} m_i=\dfrac{\Delta P}{|\dot{\omega}|_{\mathrm{d}\infty}}\,,\quad 
\sum_{i \in \mathcal{I}} d_{i} \!+\!\sum_{i \in \mathcal{G}}d_i=\dfrac{\Delta P}{\Delta\omega_{ \mathrm{d}}}\,,
\end{align}
and \eqref{eq:thm1-condition} simultaneously hold to shape the coherent frequency dynamics into a first-order one described by \eqref{eq:sys-fre} while achieving the desired magnitude of RoCoF $|\dot{\omega}|_{\mathrm{d}\infty}$ (in $\SI{}{\per\second}$) and steady-state frequency deviation $\Delta\omega_{ \mathrm{d}}$ (in $\SI{}{\pu}$), the two most important measures for frequency security. 

Since the target coherent dynamics can be guaranteed as long as \eqref{eq:thm1-condition} and \eqref{eq:I_parameters} hold, it is worth exploring how to optimally allocate frequency shaping control parameters, i.e., $m_{i}$, $d_{i}$, $\rho_{i}$, and $\sigma_i$, $\forall i \in \mathcal{I}$, to minimize the control effort needed from generators and inverters during a transient following net power imbalance of magnitude $\Delta P$, i.e.,
\begin{align}\label{eq:obj-up}
    C_{\infty,\Delta P}\!:=&\!\sum_{i \in \mathcal{N}}\int_{0}^\infty\!\left[-(u_{\mathrm{P},i}-u_{\mathrm{P},i}(\infty))\right]^2\  \mathrm{d}t\nonumber\\=&\!\int_{0}^\infty \!\!\left(\boldsymbol{u}_\mathrm{P}-\boldsymbol{u}_\mathrm{P}(\infty)\right)^T\!\!\left(\boldsymbol{u}_\mathrm{P}-\boldsymbol{u}_\mathrm{P}(\infty)\right) \mathrm{d}t
\end{align}
with $\boldsymbol{u}_\mathrm{P} := \left(u_{\mathrm{P},i}, i \in \mathcal{N} \right) =\boldsymbol{p}_\mathrm{in}-\boldsymbol{p}_\mathrm{e}$ as illustrated in Fig.~\ref{fig:model}, when these constraints are satisfied. Note that $C_{\infty,\Delta P}$ defined above captures the control effort throughout the post-fault transient by excluding the effect of steady-state control power output. This ensures that $C_{\infty,\Delta P}$ serves as a well-defined metric for sizing the energy capacity of storage units required to execute primary frequency control~\cite{knap2015sizing,jiangtps2021} by preventing it from being dominated by $\boldsymbol{u}_\mathrm{P}(\infty)$. 

However, neither the objective function \eqref{eq:obj-up} nor the $s$-domain constraint \eqref{eq:thm1-condition} is easy to address directly in optimization. Thus, in the remainder of this manuscript, we will first explicitly formulate this as a well-understood constrained optimization problem and then solve it using the existing successive convexification algorithm.


\section{Reformulation of Optimal Grid-Forming Frequency Shaping Control Allocation}

\subsection{Objective Function}
We now approximate the objective function $C_{\infty,\Delta P}$ in \eqref{eq:obj-up} through the scaled $\mathcal{H}_2$ norm of a suitably defined system derived from the system in Fig.~\ref{fig:model}. Loosely speaking, the $\mathcal{H}_2$ norm is a norm of a system response resulting
from certain input signals~\cite{Tegling2015tcns}, which can yield various physical interpretations depending on the choice of the input and output signals. For example, in the context of power systems, it has been used to quantify power losses~\cite{Tegling2015tcns}, noise sensitivity~\cite{y2017cdc, Weitenberg2018TAC, jiang2021tac}, transient stability~\cite{Poolla2019tps}, and so on. Particularly, when a system is perturbed by a unit-impulse $\delta(t)$ at each input channel individually, the squared $\mathcal{H}_2$ norm can be interpreted as the sum of the squared $\mathcal{L}_2$ norms of resulting outputs~\cite{Tegling2015tcns}. However, the system in Fig.~\ref{fig:model} has input signals $\boldsymbol{p}_\mathrm{in}$ modeled as step changes to capture the major contingencies, which precludes the direct application of the preceding interpretation. This motivates us to transform the original system into a system, denoted by $G$, with input signals $\dot{\boldsymbol{p}}_\mathrm{in}$ and output signals $\left(\boldsymbol{u}_\mathrm{P}-\boldsymbol{u}_\mathrm{P}(\infty)\right)$, which makes it legitimate to quantify $C_{\infty,\Delta P}$ with the $\mathcal{H}_2$ norm since $\boldsymbol{p}_\mathrm{in} =\boldsymbol{u}_0 \mathds{U}_{ t \geq 0 }$ means that $\dot{\boldsymbol{p}}_\mathrm{in}=\boldsymbol{u}_0 \delta(t)$. Thus, the theorem below constructs such a system $G$, upon which $n$ impulse perturbation experiments can be conducted to measure $C_{\infty,\Delta P}$. More precisely, at the $i$-th experiment, apply $\dot{\boldsymbol{p}}_\mathrm{in}=-\Delta P\boldsymbol{e}_i \delta(t)$ to system $G$, where $\boldsymbol{e}_i\in\real^n$ denotes the $i$-th standard basis vector. Then the average of the squared $\mathcal{L}_2$ norms of the resulting outputs $\left(\boldsymbol{u}_\mathrm{P}-\boldsymbol{u}_\mathrm{P}(\infty)\right)$ over all $n$ experiments is deemed an estimation of $C_{\infty,\Delta P}$ following a net power imbalance of magnitude $\Delta P$. This thought lays the foundation for the following theorem.

\begin{thm}[Control effort approximated by $\mathcal{H}_2$ norm]\label{thm:cost} The transient control effort defined in \eqref{eq:obj-up} for the system in Fig.~\ref{fig:model} is approximately given by
\begin{align}
    C_{\infty,\Delta P}\approx&\ \dfrac{(\Delta P)^2}{n}\|G\|^2_{\mathcal{H}_2}=\dfrac{(\Delta P)^2}{n}\tr{\boldsymbol{B}^T\boldsymbol{X}\boldsymbol{B}}\,,\label{eq:h2-obj}
\end{align}
where $G$ denotes the following system:
\begin{subequations}\label{eq:G-ss}
\begin{align}
    &\!\!\begin{bmatrix}
        \dot{\boldsymbol{\omega}}\\\ddot{\boldsymbol{\omega}}\\
        \dot{\boldsymbol{q}}
    \end{bmatrix}\!\!=\!\underbrace{\begin{bmatrix}
        \mathbbold{0}_{n\times n}&\!\!\boldsymbol{I}_{n}&\!\!\mathbbold{0}_{n\times n}\\-\boldsymbol{M}^{-1}(\boldsymbol{T}^{-1}\!\boldsymbol{R}\!+\!\boldsymbol{L}_\mathrm{B})&\!\!-\boldsymbol{M}^{-1}\!\boldsymbol{D}&\!\!-(\boldsymbol{M}\boldsymbol{T})^{-1}\\-\boldsymbol{T}^{-1}\!\boldsymbol{R}&\mathbbold{0}_{n\times n}&-\boldsymbol{T}^{-1}
    \end{bmatrix}}_{=:\boldsymbol{A}}\!\!\!\begin{bmatrix}
      \boldsymbol{\omega}\\ \dot{\boldsymbol{\omega}}\\ \boldsymbol{q}
    \end{bmatrix}\nonumber\\&\qquad\ \ +\!\underbrace{\begin{bmatrix}
        \mathbbold{0}_{n\times n}\\\boldsymbol{M}^{-1}\\\mathbbold{0}_{n\times n}
    \end{bmatrix}}_{=:\boldsymbol{B}}\dot{\boldsymbol{p}}_\mathrm{in}\,,\label{eq:x-ss}\\
    &\boldsymbol{u}_\mathrm{P}-\boldsymbol{u}_\mathrm{P}(\infty)\nonumber\\&=\!\underbrace{\left[\boldsymbol{I}_{n}-\dfrac{ (\boldsymbol{D}+\boldsymbol{R})\mathbbold{1}_{n}\mathbbold{1}_{n}^T}{\mathbbold{1}_{n}^T(\boldsymbol{D}\textcolor{blue}{+\boldsymbol{R}})\mathbbold{1}_{n}}\right]\!\begin{bmatrix}
        \boldsymbol{D}&\boldsymbol{M}&-\boldsymbol{I}_{n}
    \end{bmatrix}}_{=:\boldsymbol{C}}\!\!\begin{bmatrix}
      \boldsymbol{\omega}\\ \dot{\boldsymbol{\omega}}\\ \boldsymbol{q}
    \end{bmatrix}\!\,,\label{eq:up-G}
\end{align} 
\end{subequations}
and $\boldsymbol{X}$ is the observability Gramian that can be uniquely solved from the following Lyapunov equation with an additional constraint via $\boldsymbol{v}:=\begin{bmatrix}
\mathbbold{1}_{n}^T&\mathbbold{0}_{n}^T&-(\boldsymbol{R}\mathbbold{1}_{n})^T\end{bmatrix}^T$:
\begin{subequations}\label{eq:X-lyapunov-con}
\begin{align}
    &\boldsymbol{A}^T \boldsymbol{X} + \boldsymbol{X}\boldsymbol{A} = -\boldsymbol{C}^T\boldsymbol{C}\,,\label{eq:Lyap}\\
    &\boldsymbol{X}\boldsymbol{v}=\mathbbold{0}_{3n}\,.\label{eq:v}
\end{align}    
\end{subequations}
Here, $\boldsymbol{M} := \diag{m_i, i\!\in\!\mathcal{N}}\!\in\!\real^{n \times n}$, $\boldsymbol{D} := \diag{d_{i}, i \!\in\! \mathcal{N}} \in\real^{n \times n}$, and 
\begin{align*}
\boldsymbol{T}\!:=&\!\begin{bmatrix}
        \diag{\tau_{i}, i \!\in \!\mathcal{G}}& \mathbbold{0}_{|\mathcal{G}|\times|\mathcal{I}|}\\\mathbbold{0}_{|\mathcal{I}|\times|\mathcal{G}|}&\diag{\sigma_{i}, i \!\in \!\mathcal{I}}
    \end{bmatrix}\!\,,\\
    \boldsymbol{R}\!:=&{\setlength{\arraycolsep}{2pt}\begin{bmatrix}
        \diag{{\alpha_{\mathrm{t},i}}, i \!\in \!\mathcal{G}}& \mathbbold{0}_{|\mathcal{G}|\times|\mathcal{I}|}\\\mathbbold{0}_{|\mathcal{I}|\times|\mathcal{G}|}&-\diag{\rho_{i}, i \!\in \!\mathcal{I}}
    \end{bmatrix}} \!\,,\  \boldsymbol{q}\!:=\!\!\begin{bmatrix}
       \left(q_{\mathrm{t},i}, i \!\in \!\mathcal{G} \right) \\
       \left(q_{\mathrm{r},i}, i\! \in \!\mathcal{I} \right)
    \end{bmatrix}
\end{align*}
with $q_{\mathrm{r},i}$, $\forall i\in \mathcal{I}$, denoting an internal state of grid-forming shaping control.
\end{thm}

\begin{proof} 
We start by transforming the system in Fig.~\ref{fig:model} into a system with input signals $\dot{\boldsymbol{p}}_\mathrm{in}$ and output signals $\left(\boldsymbol{u}_\mathrm{P}-\boldsymbol{u}_\mathrm{P}(\infty)\right)$, which enables us to quantify control effort in transient response to step power disturbances through $\mathcal{H}_2$ norm.

The generator dynamics \eqref{eq:gen-ode} and the analogous time-domain version of inverter dynamics \eqref{eq:hi-inverter} under grid-forming frequency shaping control on each bus can be stacked as
\begin{subequations}\label{eq:sys-dyn-vec}
\begin{align}
    \boldsymbol{M} \dot{\boldsymbol{\omega}} =& -
    \boldsymbol{D} \boldsymbol{\omega} + \boldsymbol{q} + \boldsymbol{u}_\mathrm{P}\,,\label{eq:swing}\\
     \boldsymbol{T}\dot{\boldsymbol{q}}=&-\boldsymbol{R}\boldsymbol{\omega}-\boldsymbol{q}\,.\label{eq:gover}\end{align}
\end{subequations}
Taking the time derivative to \eqref{eq:swing} and then substituting \eqref{eq:gover} and the fact that $\dot{\boldsymbol{u}}_\mathrm{P} =\dot{\boldsymbol{p}}_\mathrm{in}-\dot{\boldsymbol{p}}_\mathrm{e}$ into it yields
\begin{align}\label{eq:ddomega}
    \boldsymbol{M} \ddot{\boldsymbol{\omega}} =& -
    \boldsymbol{D} \dot{\boldsymbol{\omega}} + \boldsymbol{T}^{-1}(-\boldsymbol{R} \boldsymbol{\omega}-\boldsymbol{q}) + \dot{\boldsymbol{p}}_\mathrm{in}-\dot{\boldsymbol{p}}_\mathrm{e}\,,
\end{align}
where $\dot{\boldsymbol{p}}_\mathrm{e}=\boldsymbol{L}_\mathrm{B}\boldsymbol{\omega}$ by \eqref{eq:N}. Thus, \eqref{eq:ddomega} can be written as 
\begin{align*}
    \boldsymbol{M} \ddot{\boldsymbol{\omega}} =& -
    \boldsymbol{D} \dot{\boldsymbol{\omega}} - \boldsymbol{T}^{-1}(\boldsymbol{R} \boldsymbol{\omega}+\boldsymbol{q}) + \dot{\boldsymbol{p}}_\mathrm{in}-\boldsymbol{L}_\mathrm{B}\boldsymbol{\omega}\,,
\end{align*}
which together with \eqref{eq:gover} can be turned into the standard state-space form \eqref{eq:x-ss}. 

To express $\left(\boldsymbol{u}_\mathrm{P}-\boldsymbol{u}_\mathrm{P}(\infty)\right)$ in terms of states, we first notice from \eqref{eq:swing} that 
\begin{align}\label{eq:up-ss}
    \boldsymbol{u}_\mathrm{P}=\boldsymbol{D} \boldsymbol{\omega}+\boldsymbol{M} \dot{\boldsymbol{\omega}} - \boldsymbol{q}
\end{align}
and then investigate the equilibria of \eqref{eq:x-ss} that clearly satisfy
\begin{subequations}
\begin{align}
&\mathbbold{0}_{n}=\ \dot{\boldsymbol{\omega}}(\infty)\,,\label{eq:domega_inf}\\
&\mathbbold{0}_{n}=(\boldsymbol{T}^{-1}\!\boldsymbol{R}\!+\!\boldsymbol{L}_\mathrm{B})\boldsymbol{\omega}(\infty)+\!\boldsymbol{D}\dot{\boldsymbol{\omega}}(\infty)+\boldsymbol{T}^{-1}\boldsymbol{q}(\infty)\,,\label{eq:ddomega_inf}\\
    &\boldsymbol{q}(\infty)=-\boldsymbol{R}\boldsymbol{\omega}(\infty)\,.\label{eq:gover-ss}
\end{align}    
\end{subequations}
Substituting \eqref{eq:domega_inf} and \eqref{eq:gover-ss} to \eqref{eq:ddomega_inf}, we get $\boldsymbol{L}_\mathrm{B}\boldsymbol{\omega}(\infty)=\mathbbold{0}_{n}$, which implies that 
\begin{align}
\boldsymbol{\omega}(\infty)=\omega_\mathrm{syn}\mathbbold{1}_{n}   \label{eq:omega-ss} 
\end{align}
for some synchronous frequency $\omega_\mathrm{syn}\in \real$. Combining \eqref{eq:omega-ss} and \eqref{eq:gover-ss} with the steady state of \eqref{eq:swing}, we obtain 
\begin{align}\label{eq:up-wsyn}
\boldsymbol{u}_\mathrm{P}(\infty)= \boldsymbol{D} \boldsymbol{\omega}(\infty) - \boldsymbol{q}(\infty)=\omega_\mathrm{syn}(\boldsymbol{D}+\boldsymbol{R})\mathbbold{1}_{n}\,. 
\end{align}
To solve for $\omega_\mathrm{syn}$, we observe that, $\forall t\geq 0$, 
\begin{align}\label{eq:dp-up}
\mathbbold{1}_{n}^T\boldsymbol{u}_\mathrm{P}=&\ \mathbbold{1}_{n}^T(\boldsymbol{p}_\mathrm{in}-\boldsymbol{p}_\mathrm{e})=\mathbbold{1}_{n}^T\boldsymbol{u}_0-\mathbbold{1}_{n}^T \boldsymbol{L}_\mathrm{B} \int\boldsymbol{\omega} \mathrm{d}t\nonumber\\=&-\Delta P-\mathbbold{0}_{n}^T\int\boldsymbol{\omega}\mathrm{d}t=-\Delta P\,.    
\end{align}
This combined with \eqref{eq:up-wsyn} gives $\mathbbold{1}_{n}^T\omega_\mathrm{syn}(\boldsymbol{D}+\boldsymbol{R})\mathbbold{1}_{n}=\mathbbold{1}_{n}^T\boldsymbol{u}_\mathrm{P}(\infty)=-\Delta P$, from which we can solve for
\begin{align}\label{eq:wsyn-exp}
    \omega_\mathrm{syn}=-\dfrac{\Delta P}{\mathbbold{1}_{n}^T(\boldsymbol{D}+\boldsymbol{R})\mathbbold{1}_{n}}\,.
\end{align}
\textcolor{blue}{Substituting \eqref{eq:wsyn-exp} to \eqref{eq:up-wsyn} yields 
\begin{align}\label{eq:up-dp}
\boldsymbol{u}_\mathrm{P}(\infty)= -\dfrac{ (\boldsymbol{D}+\boldsymbol{R})\mathbbold{1}_{n}\Delta P}{\mathbbold{1}_{n}^T(\boldsymbol{D}+\boldsymbol{R})\mathbbold{1}_{n}}\,. 
\end{align}
The output expression \eqref{eq:up-G} in terms of states is a direct application of \eqref{eq:dp-up} and \eqref{eq:up-ss} to \eqref{eq:up-dp}, i.e., 
\begin{align*}
\boldsymbol{u}_\mathrm{P}-\boldsymbol{u}_\mathrm{P}(\infty)=&\ \boldsymbol{u}_\mathrm{P}-\dfrac{ (\boldsymbol{D}+\boldsymbol{R})\mathbbold{1}_{n}\mathbbold{1}_{n}^T\boldsymbol{u}_\mathrm{P}}{\mathbbold{1}_{n}^T(\boldsymbol{D}+\boldsymbol{R})\mathbbold{1}_{n}}\\=&\left[\boldsymbol{I}_{n}\!-\dfrac{ (\boldsymbol{D}+\boldsymbol{R})\mathbbold{1}_{n}\mathbbold{1}_{n}^T}{\mathbbold{1}_{n}^T(\boldsymbol{D}+\boldsymbol{R})\mathbbold{1}_{n}}\right]\!\left(\boldsymbol{D} \boldsymbol{\omega}+\boldsymbol{M} \dot{\boldsymbol{\omega}} - \boldsymbol{q}\right)\,. 
\end{align*}}

Clearly, the system $G$ in \eqref{eq:G-ss} exhibits the desired form with input $\dot{\boldsymbol{p}}_\mathrm{in}$ and output $\left(\boldsymbol{u}_\mathrm{P}-\boldsymbol{u}_\mathrm{P}(\infty)\right)$. 
We are now ready to connect the control effort $C_{\infty,\Delta P}$ with $\|G\|^2_{\mathcal{H}_2}$. Note that $\|G\|^2_{\mathcal{H}_2}$ is equal to the sum of the squared $\mathcal{L}_2$ norms of the resulting $\left(\boldsymbol{u}_\mathrm{P}-\boldsymbol{u}_\mathrm{P}(\infty)\right)$ over all $n$ experiments, where in each the
system $G$ is fed a unit-impulse at the $i$-th input channel, i.e., $\dot{\boldsymbol{p}}_\mathrm{in}=\boldsymbol{e}_i \delta(t)$~\cite{Tegling2015tcns}. Hence, when we customize the experiments by scaling the impulse in each experiment by $\Delta P$, it is easy to see the squared $\mathcal{L}_2$ norms of $\left(\boldsymbol{u}_\mathrm{P}-\boldsymbol{u}_\mathrm{P}(\infty)\right)$ in each experiment is scaled by $(\Delta P)^2$, which results in the coefficient $(\Delta P)^2/n$ before $\|G\|^2_{\mathcal{H}_2}$ in \eqref{eq:h2-obj} as we use the mean of these $n$ customized experiments to approximate $C_{\infty,\Delta P}$.

The proof of the formula for calculating $\|G\|^2_{\mathcal{H}_2}$ in \eqref{eq:h2-obj}, together with constraints \eqref{eq:X-lyapunov-con}, follows a similar argument to that for \cite[Lemma~1]{Poolla2017TAC}. The key idea is summarized below for completeness. First, it is a standard result that $\|G\|^2_{\mathcal{H}_2}:=\sum_{i \in \mathcal{N}}\|G(t)\boldsymbol{e}_i\|^2_{\mathcal{L}_2}=\int_0^{\infty}\tr{G(t)^TG(t)} \mathrm{d} t=\tr{\boldsymbol{B}^T\tilde {\boldsymbol{X}}\boldsymbol{B}}$ with $\tilde {\boldsymbol{X}}:=\int_0^{\infty}e^{\boldsymbol{A}^Tt}\boldsymbol{C}^T\boldsymbol{C}e^{\boldsymbol{A}t}\mathrm{d} t$~\cite[Section 2.3.1]{Jiang2021jhu}. For the simple case where $\boldsymbol{A}$ is asymptotically stable, $\tilde {\boldsymbol{X}}$ is the unique solution to the Lyapunov equation~\eqref{eq:Lyap}. However, in our setting, it is not hard to show that $\rank{\boldsymbol{A}^T}=\rank{\boldsymbol{A}}=3n-1$. More precisely, $\boldsymbol{A}$ has an unobservable $0$ eigenvalue associated with eigenvector $\boldsymbol{v}$,\footnote{\textcolor{blue}{Our parameter requirement guarantees that all nonzero eigenvalues of $\boldsymbol{A}$ lie in  open left half-plane.}} i.e., $\boldsymbol{C}e^{\boldsymbol{A}t}\boldsymbol{v}=\boldsymbol{C}\sum_{k=0}^{\infty}t^k\boldsymbol{A}^k/(k!)\boldsymbol{v}=\boldsymbol{C}\boldsymbol{v}=\mathbbold{0}_{n}$, $\forall t\geq 0$. This leads to the issue that~\eqref{eq:Lyap} has a solution family parametrized by $\boldsymbol{X}(a):=\tilde {\boldsymbol{X}}+a\boldsymbol{\zeta}\boldsymbol{\zeta}^T$ for $a\in\real$ with $\boldsymbol{\zeta}:=\begin{bmatrix}
(\boldsymbol{D}\mathbbold{1}_{n})^T&(\boldsymbol{M}\mathbbold{1}_{n})^T&-\mathbbold{1}_{n}^T\end{bmatrix}^T$ being the kernel of $\boldsymbol{A}^T$, i.e., $\boldsymbol{A}^T\boldsymbol{\zeta}=\mathbbold{0}_{3n}$. Finally, by imposing an additional constraint~\eqref{eq:v}, one can restrict the solution exactly to the desired $\tilde {\boldsymbol{X}}$. To see this, we note that $\tilde {\boldsymbol{X}}$ inherently satisfy ~\eqref{eq:v} since $\tilde {\boldsymbol{X}}\boldsymbol{v}=\int_0^{\infty}e^{\boldsymbol{A}^Tt}\boldsymbol{C}^T\boldsymbol{C}e^{\boldsymbol{A}t}\boldsymbol{v}\mathrm{d} t=\mathbbold{0}_{3n}$, while $\boldsymbol{X}(a)$ satisfies~\eqref{eq:v} only if $\mathbbold{0}_{3n}=(\tilde {\boldsymbol{X}}+a\boldsymbol{\zeta}\boldsymbol{\zeta}^T)\boldsymbol{v}=\tilde {\boldsymbol{X}}\boldsymbol{v}+a\boldsymbol{\zeta}\boldsymbol{\zeta}^T\boldsymbol{v}=a\boldsymbol{\zeta}(\mathbbold{1}_{n}^T\boldsymbol{D}\mathbbold{1}_{n}+\mathbbold{1}_{n}^T\boldsymbol{R}\mathbbold{1}_{n})=a\boldsymbol{\zeta}(\sum_{i \in \mathcal{N}} d_i+0)=a(\sum_{i \in \mathcal{N}} d_i) \boldsymbol{\zeta}$. This implies that $a=0$ must hold since $\sum_{i \in \mathcal{N}} d_i>0$, which ensures that the only solution to \eqref{eq:X-lyapunov-con} is $\boldsymbol{X}(0)=\tilde {\boldsymbol{X}}$.
\end{proof}
\begin{rem}[Positive objective scaling invariance]\label{rem:co-cost}
The positive constant coefficient $(\Delta P)^2/n$ can be omitted from the objective function without changing optimizers, even if $\Delta P$ also appears in constraints. To see this, consider the constrained optimization problem $\min_{\boldsymbol{x}\in\mathcal{F}(\Delta P)}g(\Delta P)f(\boldsymbol{x})$, where $f(\boldsymbol{x})$ represents the non-scaled objective function, $\mathcal{F}(\Delta P)$ denotes the feasible set given $\Delta P$, and $g(\Delta P)>0$ is a positive scaling factor for any given $\Delta P>0$. Let $\mathcal{S}_1:=\arg \min_{\boldsymbol{x}\in\mathcal{F}(\Delta P)}g(\Delta P)f(\boldsymbol{x})$ and $\mathcal{S}_2:=\arg \min_{\boldsymbol{x}\in\mathcal{F}(\Delta P)}f(\boldsymbol{x})$. We can prove $\mathcal{S}_1=\mathcal{S}_2$ by showing $\mathcal{S}_1\subseteq\mathcal{S}_2$ and $\mathcal{S}_2\subseteq\mathcal{S}_1$. First, $\forall \boldsymbol{x}^*\in\mathcal{S}_1$, it holds that $g(\Delta P)f(\boldsymbol{x}^*)\leq g(\Delta P)f(\boldsymbol{x})$, $\forall \boldsymbol{x}\in\mathcal{F}(\Delta P)$. Dividing both sides by $g(\Delta P)>0$ preserves the inequality, which gives $f(\boldsymbol{x}^*)\leq f(\boldsymbol{x})$, $\forall \boldsymbol{x}\in\mathcal{F}(\Delta P)$. Thus, $\boldsymbol{x}^*\in\mathcal{S}_2$, which implies $\mathcal{S}_1\subseteq\mathcal{S}_2$. Similarly, one can show $\mathcal{S}_2\subseteq\mathcal{S}_1$.
\end{rem}

\subsection{$s$-Domain Constraint}
The $s$-domain constraint \eqref{eq:thm1-condition} required by shaping control plays a pivotal role in removing the coherent frequency Nadir~\cite{jiang2021lcss}. However, it is not straightforward how to handle it in optimization, which motivates us to simplify it into a standard equality constraint via the following theorem. 
\begin{thm}[Simplification of $s$-domain constraint]\label{thm:s-const} The $s$-domain constraint \eqref{eq:thm1-condition} is equivalent to
\begin{align}\label{eq:thm1-condition-quad}
    \boldsymbol{\rho}^{T} \boldsymbol{\Sigma}(\boldsymbol{\sigma}) \boldsymbol{\rho}
+
\boldsymbol{b}(\boldsymbol{\sigma})^{T}\boldsymbol{\rho}
+
c = 0\,,
\end{align}
where $\boldsymbol{\rho}:=\left(\rho_i, i \in \mathcal{I} \right) \in \real^{|\mathcal{I}|}$, $\boldsymbol{\sigma}:=\left(\sigma_i, i \in \mathcal{I} \right) \in \real^{|\mathcal{I}|}$, $\boldsymbol{\Sigma}(\boldsymbol{\sigma})\in \real^{|\mathcal{I}|\times |\mathcal{I}|}$ is a positive semi-definite matrix with \begin{align}\label{eq:Sig_ij}
    \Sigma_{ij}:= \frac{1}{
\sigma_i + \sigma_j
}\,,
\end{align}
and 
\begin{align*}
    \boldsymbol{b}(\boldsymbol{\sigma})\!:=\!\!\left(\!-2\!\sum_{j \in \mathcal{G}}\frac{ \alpha_{\mathrm{t},j}}{\sigma_i\! +\!\tau_j}, i \!\in \!\mathcal{I}\! \right) \!\!\in \!\real^{|\mathcal{I}|}\!\,,\ c\!:=\!\!\sum_{i \in \mathcal{G}}\sum_{j \in \mathcal{G}}\frac{\alpha_{\mathrm{t},i} \alpha_{\mathrm{t},j}}{\tau_i \!+\!\tau_j}\,.
\end{align*}
\end{thm}
\begin{proof}
Taking the inverse Laplace transform of~\eqref{eq:thm1-condition} yields
\begin{align}\label{eq:thm1-condition-t}
0=\sum_{i \in \mathcal{I}} \frac{\rho_{i}}{\sigma_i}e^{-\frac{t}{\sigma_i}} - \sum_{i \in \mathcal{G}}\frac{\alpha_{\mathrm{t},i}}{\tau_i}e^{-\frac{t}{\tau_i}}\,.
\end{align}
Then, integrating the squared \eqref{eq:thm1-condition-t} over $t\in[0,\infty)$ gives
\begin{align}
    0=&\!\int_{0}^{\infty} \!\left(\sum_{i \in \mathcal{I}} \frac{\rho_{i}}{\sigma_i}e^{-\frac{t}{\sigma_i}} - \sum_{i \in \mathcal{G}}\frac{\alpha_{\mathrm{t},i}}{\tau_i}e^{-\frac{t}{\tau_i}}\right)^2\mathrm{d}t\label{eq:square-constraint}\\
    =&\!\int_{0}^{\infty} \!\Bigg[\!\!\left(\sum_{i \in \mathcal{I}} \!\frac{\rho_{i}}{\sigma_i}e^{-\frac{t}{\sigma_i}}\! \right)^2\!\!-\!2\!\left(\sum_{i \in \mathcal{I}} \!\frac{\rho_{i}}{\sigma_i}e^{-\frac{t}{\sigma_i}} \right)\!\!\left(\sum_{i \in \mathcal{G}}\!\frac{\alpha_{\mathrm{t},i}}{\tau_i}e^{-\frac{t}{\tau_i}}\right)\nonumber\\&\qquad+ \left(\sum_{i \in \mathcal{G}}\frac{\alpha_{\mathrm{t},i}}{\tau_i}e^{-\frac{t}{\tau_i}}\!\right)^2\Bigg]\mathrm{d}t
    \nonumber\\
    =&\!\int_{0}^{\infty} \!\Bigg[\sum_{i \in \mathcal{I}}\sum_{j \in \mathcal{I}}\left( \!\frac{\rho_{i}}{\sigma_i}e^{-\frac{t}{\sigma_i}}\! \right)\!\!\left( \!\frac{\rho_{j}}{\sigma_j}e^{-\frac{t}{\sigma_j}}\! \right)\!\!\nonumber\\&\qquad-2\sum_{i \in \mathcal{I}}\sum_{j \in \mathcal{G}}\left( \!\frac{\rho_{i}}{\sigma_i}e^{-\frac{t}{\sigma_i}} \right)\!\!\left(\!\frac{\alpha_{\mathrm{t},j}}{\tau_j}e^{-\frac{t}{\tau_j}}\right)\nonumber\\&\qquad+ \sum_{i \in \mathcal{G}}\sum_{j \in \mathcal{G}}\left(\frac{\alpha_{\mathrm{t},i}}{\tau_i}e^{-\frac{t}{\tau_i}}\!\right)\!\!\left(\frac{\alpha_{\mathrm{t},j}}{\tau_j}e^{-\frac{t}{\tau_j}}\!\right)\Bigg]\mathrm{d}t
    \nonumber\\
    =&\sum_{i \in \mathcal{I}}\sum_{j \in \mathcal{I}}\frac{\rho_{i} \rho_{j}}{\sigma_i \sigma_j}\int_{0}^{\infty}\!\!e^{-\left(\frac{1}{\sigma_i}+\frac{1}{\sigma_j}\right)t} \ \mathrm{d}t\!\!\nonumber\\&-2\sum_{i \in \mathcal{I}}\sum_{j \in \mathcal{G}}\frac{\rho_{i} \alpha_{\mathrm{t},j}}{\sigma_i \tau_j}\int_{0}^{\infty}\!\!e^{-\left(\frac{1}{\sigma_i}+\frac{1}{\tau_j}\right)t} \ \mathrm{d}t\nonumber\\&+ \sum_{i \in \mathcal{G}}\sum_{j \in \mathcal{G}}\frac{\alpha_{\mathrm{t},i} \alpha_{\mathrm{t},j}}{\tau_i \tau_j}\int_{0}^{\infty}\!\!e^{-\left(\frac{1}{\tau_i}+\frac{1}{\tau_j}\right)t} \ \mathrm{d}t
    \nonumber\\
    =&\sum_{i \in \mathcal{I}}\sum_{j \in \mathcal{I}}\frac{\rho_{i} \rho_{j}}{\sigma_i +\sigma_j}\!-2\sum_{i \in \mathcal{I}}\sum_{j \in \mathcal{G}}\frac{\rho_{i} \alpha_{\mathrm{t},j}}{\sigma_i +\tau_j}+ \!\sum_{i \in \mathcal{G}}\sum_{j \in \mathcal{G}}\frac{\alpha_{\mathrm{t},i} \alpha_{\mathrm{t},j}}{\tau_i +\tau_j}\,,\nonumber
\end{align}
whose compact form is exactly~\eqref{eq:thm1-condition-quad}. 

The positive semi-definite property of $\boldsymbol{\Sigma}(\boldsymbol{\sigma})$ follows from the following two facts. First, $\boldsymbol{\Sigma}(\boldsymbol{\sigma})$ is symmetric since $\Sigma_{ji}=\Sigma_{ji}$ holds trivially, $\forall i,j \in \mathcal{I}$, by its construction \eqref{eq:Sig_ij}. Second, observe from the first summation term along the derivations in \eqref{eq:square-constraint} that, $\forall \boldsymbol{\rho} \in \real^{|\mathcal{I}|}$,  
\begin{align*}
    \boldsymbol{\rho}^{T} \boldsymbol{\Sigma}(\boldsymbol{\sigma}) \boldsymbol{\rho}\!=\!\sum_{i \in \mathcal{I}}\sum_{j \in \mathcal{I}}\frac{\rho_{i} \rho_{j}}{\sigma_i +\sigma_j}=\!\!\int_{0}^{\infty} \!\left(\sum_{i \in \mathcal{I}} \!\frac{\rho_{i}}{\sigma_i}e^{-\frac{t}{\sigma_i}}\! \right)^2\mathrm{d}t\geq 0\,.
\end{align*}
This concludes the proof.
\end{proof}
\textcolor{blue}{
\begin{rem}[Relaxation of exact matching constraint]\label{rem:s-const-ieq} In fact, it is not hard to show via the inear independence of exponential functions that \eqref{eq:thm1-condition-quad} requires rigorously individual matching between generators and inverters, which severely restricts the design space. Thus, in the remainder of this manuscript, we relax it as follows:
\begin{align}\label{eq:thm1-condition-quad-ieq}
    \boldsymbol{\rho}^{T} \boldsymbol{\Sigma}(\boldsymbol{\sigma}) \boldsymbol{\rho}
+
\boldsymbol{b}(\boldsymbol{\sigma})^{T}\boldsymbol{\rho}
+
c \leq \epsilon\,,
\end{align}   
where $\epsilon>0$ is the prescribed relaxation tolerance.
\end{rem}}
\subsection{Reformulated Optimal Allocation Problem}
After combining the previous results including Theorem~\ref{thm:cost}, \textcolor{blue}{Remarks~\ref{rem:co-cost} and~\ref{rem:s-const-ieq}}, we can reformulate the optimal allocation of grid-forming frequency shaping control as the following constrained optimization problem:
\begin{subequations}\label{eq:opt-ss}
	\begin{align}
		\label{eq:edp2.a}
		\min_{m_{i}>0, d_{i}>\rho_{i}\geq0, \sigma_i>0, \forall i \in \mathcal{I}, \boldsymbol{X}=\boldsymbol{X}^T} \ \      \tr{\boldsymbol{B}^T\boldsymbol{X}\boldsymbol{B}}    \\
		\label{eq:edp2.b}
		\mathrm{s.t.} \qquad\quad\ \   
		\eqref{eq:I_parameters}, \eqref{eq:X-lyapunov-con}, \text{and}\ \textcolor{blue}{\eqref{eq:thm1-condition-quad-ieq}}\,,
	\end{align}	\label{eq:edp2}%
\end{subequations}
whose objective and constraints both have non-convexity.

\section{Solution Method and Case Study}
In this section, we address the non-convex optimization problem \eqref{eq:opt-ss} using the existing successive convexification algorithm~\cite{Oguri2023cdc} and present the effectiveness of the optimal allocation of grid-forming frequency shaping control through simulations on the Icelandic Power Network~\cite{iceland}.

\subsection{Solution Method: Successive Convexification}\label{ssec:-suc-conv}
To solve \eqref{eq:opt-ss} using the successive convexification algorithm proposed by~\cite[Algorithm~1]{Oguri2023cdc}, we need to turn \eqref{eq:opt-ss} into the standard convexified optimization problem therein~\cite[Problem 4]{Oguri2023cdc}. We begin by rewriting \eqref{eq:opt-ss} as the following non-convex penalty problem that penalizes the non-convex objective function~\eqref{eq:edp2.a} as well as violations of \textcolor{blue}{the non-convex constraint \eqref{eq:thm1-condition-quad-ieq}}~\cite[Problem 3]{Oguri2023cdc}:
\textcolor{blue}{\begin{subequations}\label{eq:opt-J_pro3}
	\begin{eqnarray}
		\label{eq:J_pro3.a}
		\!\!\!\!\!\!\!\!\!\!\!\!\!\!\min_{\boldsymbol{z}} &&       \!\!\!\!\!\!\!\!\!\!\!J(\boldsymbol {z})
\!:=
l^2 \!+ \!\lambda g(\boldsymbol{z}) \!+ \!\dfrac{w}{2}
g(\boldsymbol z)^2\!+\!\mu h(\boldsymbol z) \!+\!\dfrac{w}{2}
[h(\boldsymbol z)]_{+}^2 \\
		\label{eq:J_pro.b}
		\!\!\!\!\!\!\!\!\!\!\mathrm{s.t.} &&  
		\!\!\!\!\!\!\eqref{eq:I_parameters},\text{and}\ m_{i}>0, d_{i}>\rho_{i}\geq0, \sigma_i>0, \forall i \in \mathcal{I} \,,
	\end{eqnarray}%
\end{subequations}
where $\boldsymbol{z}
\!:=\![l\ \left(m_i, i \in \mathcal{I} \right)^T\ \left(d_i, i \in \mathcal{I} \right)^T\ \boldsymbol{\rho}^T\ \boldsymbol{\sigma}^T]^T\!\in \!\real^{4|\mathcal{I}|+1}$, $\lambda$ and $\mu$ are the Lagrange multipliers, $w$ is the penalty factor, $g(\boldsymbol{z})$ and $h(\boldsymbol{z})$ denotes the non-affine and non-convex constraint functions given by}\footnote{Here, entries of $\boldsymbol{X}$ are not considered as optimization variables any more, since $\boldsymbol{X}$ is uniquely determined for any given $\boldsymbol{z}$ by Theorem~\ref{thm:cost}. Thus, we use $\boldsymbol{X}(\boldsymbol{z})$ to emphasize that $\boldsymbol{X}$ is an implicit function of $\boldsymbol{z}$ that can be directly solved by using the \textcolor{blue}{existing routine $\operatorname{Lyap}(\boldsymbol{A},\boldsymbol{C}^T\boldsymbol{C})$ proposed by~\cite[Section~III-D]{Poolla2017TAC}. This routine directly outputs $\boldsymbol{X}(\boldsymbol{z})$ that satisfies the constrained Lyapunov equation \eqref{eq:X-lyapunov-con} before evaluating the objective function and the remaining constraints at each trial point of $\boldsymbol{z}$, which allows us to remove constraints related to \eqref{eq:X-lyapunov-con} from $\boldsymbol{g}(\boldsymbol{z})$.}}
\textcolor{blue}{\begin{align*}
g(\boldsymbol{z})
:=&\
\tr{\boldsymbol{B}^T\boldsymbol{X}(\boldsymbol{z})\boldsymbol{B}}-l^2\,,\\ h(\boldsymbol{z})
:=&\ \boldsymbol{\rho}^{T} \boldsymbol{\Sigma}(\boldsymbol{\sigma}) \boldsymbol{\rho}
+
\boldsymbol{b}(\boldsymbol{\sigma})^{T}\boldsymbol{\rho}
+
c-\epsilon\,,
\label{eq:g_definition}
\end{align*}
with $[h(\boldsymbol{z})]_{+}:=\max \{0,h(\boldsymbol{z})\}$.} The non-convex penalty problem~\eqref{eq:opt-J_pro3} can be linearized around the current solution $\bar{\boldsymbol{z}}$ at each iteration, which yields the following convexified penalty subproblem~\cite[Problem 4]{Oguri2023cdc}\textcolor{blue}{\begin{subequations}\label{eq:opt-J_pro4}
	\begin{eqnarray}
		\label{eq:J_pro4.a}
		\!\!\!\!\!\!\!\!\!\!\min_{\boldsymbol{z},\xi,\zeta} &&       \!\!\!\!\!\!\!\!L(\boldsymbol {z},\xi,\zeta)
:=
l^2 + \lambda\xi + \dfrac{w}{2}
\xi^2 +\mu \zeta+ \dfrac{w}{2}\zeta^2  \\
		\label{eq:J_pro4.b}
		\!\!\!\!\!\!\!\!\!\!\mathrm{s.t.} &&
\!\!\!\!\!\!\!\!\tilde{g}(\boldsymbol{z}):=g(\bar{\boldsymbol{z}})+\nabla_{\boldsymbol {z}}g(\bar{\boldsymbol{z}})^T
(\boldsymbol {z}-\bar{\boldsymbol{z}})
=
\xi\\&&
\!\!\!\!\!\!\!\!\tilde{h}(\boldsymbol{z}):=h(\bar{\boldsymbol{z}})+\nabla_{\boldsymbol {z}}h(\bar{\boldsymbol{z}})^T
(\boldsymbol {z}-\bar{\boldsymbol{z}})
\leq
\zeta\,,\ \zeta\geq 0\\&&
\!\!\!\!\!\!\!\!\|
\boldsymbol{z}-\bar{\boldsymbol{z}}\|_{\infty}
\le
r
\\  
		&&\!\!\!\!\!\!\!\!\eqref{eq:I_parameters},\text{and}\ m_{i}>0, d_{i}>\rho_{i}\geq0, \sigma_i>0, \forall i \in \mathcal{I} \,,
	\end{eqnarray}%
\end{subequations}
where $\xi$ and $\zeta$ are dual variables} indicating constraint violations and $r>0$ is a trust region bound for linearization. With \eqref{eq:opt-J_pro3} and \eqref{eq:opt-J_pro4} established, we are ready to utilize the successive convexification algorithm~\cite[Algorithm~1]{Oguri2023cdc} to solve the optimal allocation of grid-forming frequency shaping control problem~\eqref{eq:opt-ss}.

\subsection{Case Study: Modified Icelandic Power Network}
The simulations are conducted on the modified Icelandic Power Network~\cite{iceland}, where the network setup is similar to the one in~\cite[Section V]{jiang2021lcss} except that $9$ buses are randomly picked as inverters instead of $6$. Suppose that the maximum value of a sudden power imbalance that the
system should survive is $\Delta P=\SI{0.6}{\pu}$ \textcolor{blue}{(on \SI{100}{\mega\VA} base).
Our goal here is to allocate frequency shaping control resources to obtain desired magnitude of RoCoF $|\dot{\omega}|_{\mathrm{d}\infty}=\SI{0.00625}{\per\second}$  and steady-state frequency deviation $\Delta\omega_{ \mathrm{d}}=\SI{0.00625}{\pu}$ (on \SI{60}{\hertz} base) with minimum transient control effort. We achieve this goal through two-phase numerical experiments.}

\textcolor{blue}{In phase one, we build up a kron reduced model of the modified Icelandic Power Network adopted in ~\cite[Section V]{jiang2021lcss} with only $6$ generator buses and $9$ inverter buses being retained, which allows us to fit our problem into the optimization framework \eqref{eq:opt-ss} and solve for optimal inverter parameter values via the successive convexification procedure discussed in Section~\ref{ssec:-suc-conv}. Considering that the non-convex optimization solution highly depends on initial points, we select the initial point for successive convexification algorithm through multi-start experiments, after which the one that yields the lowest transient control effort is chosen as our initial point. During iterations, the convexified subproblems \eqref{eq:opt-J_pro4} are modeled using YALMIP and solved using MOSEK. The iteration stops either when satisfying both the optimality tolerance $\epsilon_{\mathrm{opt}}=10^{-4}$ and feasibility tolerance $\epsilon_{\mathrm{feas}}=10^{-2}$ or reaching the maximum number of allowed iterations that is $1250$ in our setting.}

\textcolor{blue}{In phase two, we test the performance of the optimal inverter allocation obtained from the kron reduced model on the more realistic model used in~\cite[Section V]{jiang2021lcss}, which includes load buses, nonlinear power flows, line losses, and voltage dynamics. To demonstrate the effectiveness of our proposed optimal inverter allocation approach, a baseline case is chosen as follows: $m_i$ is averaged across each inverter such that their total value satisfies the first constraint in \eqref{eq:I_parameters}, i.e., $m_i:=(\Delta P/|\dot{\omega}|_{\mathrm{d}\infty}-\sum_{i \in \mathcal{G}} m_i)/|\mathcal{I}|$ ; $\rho_i=\alpha_{\mathrm{t},i}$ and $\sigma_i=\tau_i$ to realize individual turbine matching\footnote{\textcolor{blue}{Notably, multiple feasible individual matching choices exist since $|\mathcal{I}| > |\mathcal{G}|$. Thus, we enumerate all possible cases and select the one that produces the smallest transient control effort as the baseline, which provides a comparatively strong baseline to ensure a fair comparison.}} proposed by \cite[Section III]{jiang2021lcss}; $d_i$ is constructed in two steps due to stability concern, where each inverter is first assigned a quantity of $\rho_i$ to satisfy the stability condition $d_i>\rho_i$ and then added by the average of the remaining damping needed from each inverter to satisfy the second constraint in \eqref{eq:I_parameters}, i.e., $d_i:=\rho_i+(\Delta P/\Delta\omega_{ \mathrm{d}}-\sum_{i \in \mathcal{G}}d_i-\sum_{i\in\mathcal{I}}\rho_i
)/|\mathcal{I}|$. The resulting baseline and optimal allocation of grid-forming frequency shaping control parameters are presented in Fig.~\ref{fig:parameter allo}. The comparison of system performance in these two cases is provided in Fig.~\ref{fig:fre-comp} and \ref{fig:control-comp}, which clearly demonstrates that both allocations have achieved the desired frequency security metrics, while the optimal allocation reduces the transient control effort at a cost of higher steady-state control effort. Thus, the optimal allocation that strikes a trade-off between transient and steady-state control effort might be a interesting future direction.}

\begin{figure}[t!]
\centering
\subfigure[Inverter parameter allocation]
{\includegraphics[width=0.89\columnwidth]{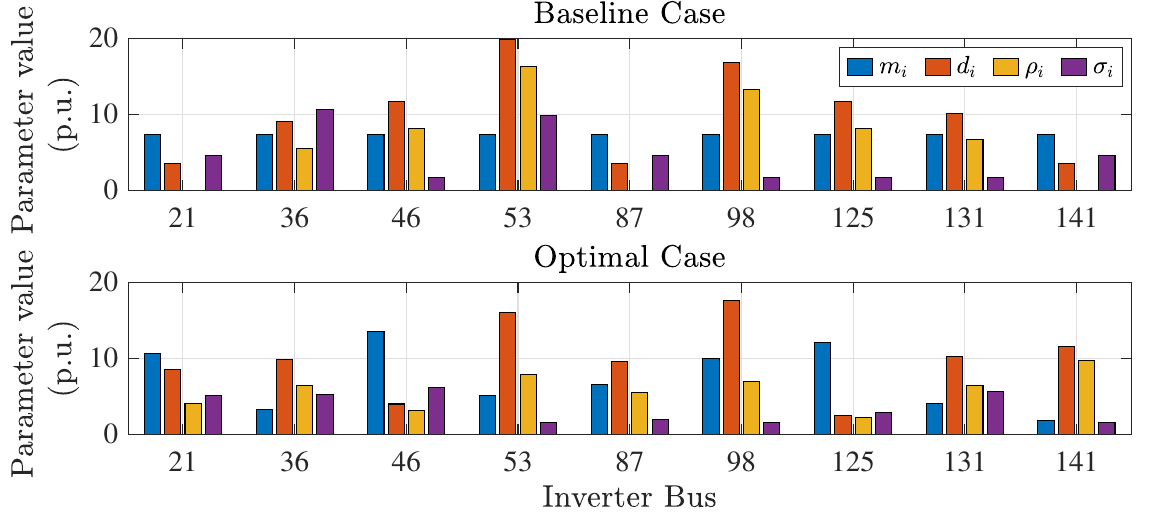}\label{fig:parameter allo}}
\hfil
\subfigure[Frequency deviations]
{\includegraphics[width=0.89\columnwidth]{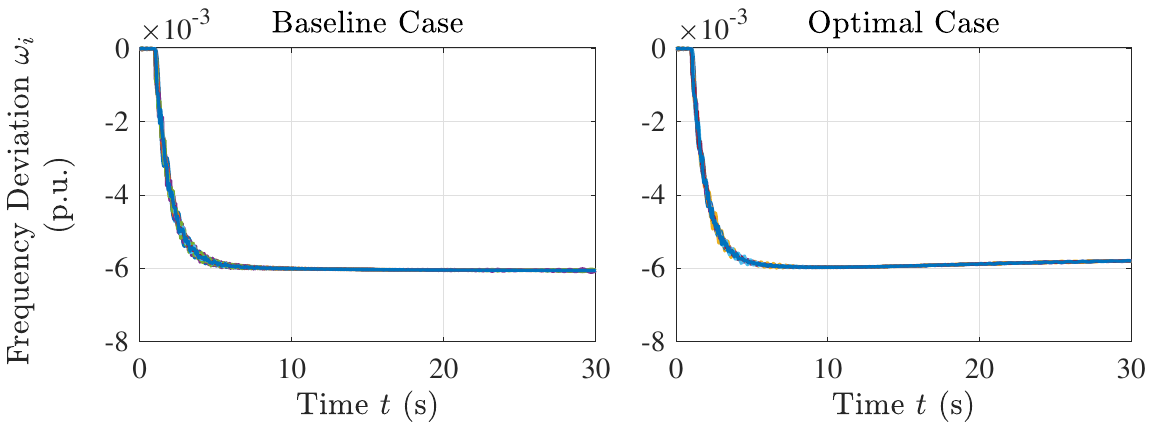}\label{fig:fre-comp}}
\subfigure[Control effort]
{\includegraphics[width=0.89\columnwidth]{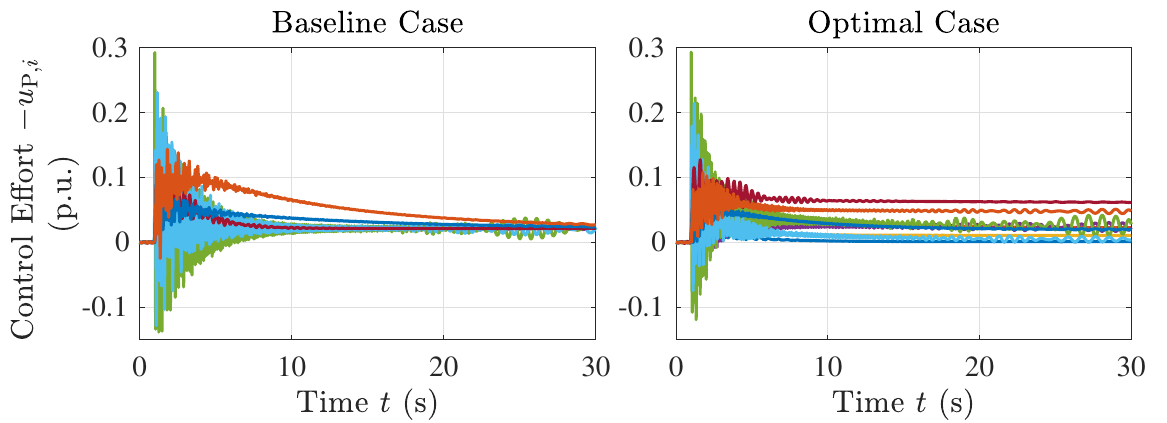}\label{fig:control-comp}}
\caption{\textcolor{blue}{Comparison between baseline and optimal allocation of grid-forming frequency shaping control when a $-0.6$ p.u. step power change is introduced to a randomly picked bus.}}
\label{fig:compareS}
\end{figure}

\section{Conclusions}
The allocation of grid-forming frequency shaping control resources to minimize the transient control effort required for achieving specified aggregate frequency dynamics has been formulated as a constrained $\mathcal{H}_2$ norm optimization problem, where attention has been paid to the transient cost quantification and $s$-domain constraint simplification. The resulting non-convex optimization problem has been solved using the existing successive convexification algorithm. \textcolor{blue}{Numerical simulations shows that the optimal allocation reduces the transient control effort at a cost of higher steady-state control effort. Thus, the optimal allocation that strikes a trade-off between transient and steady-state control effort might be a interesting future direction.}

\bibliographystyle{IEEEtran}
\bibliography{main}

\end{document}